\documentclass[a4paper,12pt]{article}
\usepackage{localeng}
\usepackage[margin=25mm]{geometry}
\DeclareMathOperator{\KS}{\mathrm{C}}
\DeclareMathOperator{\KP}{\mathrm{K}}
\newcommand{\cnd}{{\mskip 1mu | \mskip 1mu}} 
\newtheorem{proposition}{\textsc{Proposition}}
\theoremstyle{definition}
\newtheorem*{definition}{\textsc{Definition}}
\theoremstyle{remark}
\newtheorem{lemma}{\textsc{Lemma}}
\newtheorem{question}{\textsc{Question}}

\title{All Kolmogorov complexity functions are optimal,\\ but are some more optimal?}
\author{Bruno Bauwens\thanks{
    HSE university, Moscow, Russia. Supported by a grant from the Russian Science Foundation No. 20-11-20203, \protect\url{https://rscf.ru/project/20-11-20203/}
    }, 
  Alexander Kozachinskiy\thanks{
    CENIA, Santiago, Chile.
    Supported by the ANID Fondecyt Iniciación grant 11250060 and National Center for Artifcial Intelligence CENIA FB210017, Basal ANID.
    },
  Alexander Shen\thanks{LIRMM, Univ Montpellier, CNRS, Montpellier, France. Supported by ANR-21-CE48-0023 FLITTLA grant}
  }

\begin{document}
\init % Needed to keep English as a default language (WORKAROUND!)
\maketitle

\begin{flushright}
 \emph{To Yuri Gurevich on the occasion of his anniversary}
\end{flushright}

\begin{abstract}
    Kolmogorov  (1965) defined the complexity of a string $x$ as the minimal length of a program generating $x$. Obviously this definition depends on the choice of the programming language. Kolmogorov noted that there exist \emph{optimal} programming languages that make the complexity function minimal up to $O(1)$ additive terms, and we should take one of them --- but which one?
    
Is there a chance to agree on some specific programming language in this definition? Or at least should we add some other requirements to optimality? What can we achieve in this way?      

In this paper we discuss different suggestions of this type that appeared since 1965, specifically a stronger requirement of universality (and show that in many cases this does not change the set of complexity functions).

\end{abstract}

\section{Kolmogorov complexity}

Algorithmic (Kolmogorov) complexity of a finite object is defined as the minimal length of a program producing this object. Obviously this definition depends on the programming language, so we need to fix the language somehow. More formally, let $U$ be a partial computable function whose arguments and values are strings (one could consider $U$ as an interpreter of some programming language, or just as a decompressor). The complexity of a binary string $x$ with respect to $U$ is defined as
\[
\KS_U(x) = \min\{ |p|\colon U(p)=x\},
\]
where $|p|$ is the length of binary string $p$. Here the strings $p$ such that $U(p)=x$ can be called ``programs'' or ``descriptions'' for $x$.
 
Kolmogorov observed~\cite{kolmogorov1965} that there exist \emph{optimal} $U$ that make $\KS_U$ minimal up to an additive $O(1)$-term. But there are many optimal algorithms $U$ that lead to different functions $\KS_U$. They differ by an $O(1)$ additive term (due to optimality). Which one should we choose to call it \emph{the} Kolmogorov complexity function? Kolmogorov says \cite[p.~10]{kolmogorov1965} that 
\begin{quote}
  \ldots it is doubtful that such a choice can be done in a non-arbitrary way. One could think, however, that different possible ``reasonable'' candidates will lead to complexity functions that differ by hundreds of bits, not tens of thousands of bits. So the quantities like complexity of the \emph{War and Peace} text are well defined from the practical viewpoint.
\end{quote}

Kolmogorov actually defined \emph{conditional complexity} $\KS_U(x\cnd y)$. In this setting $U$ has two arguments: the first argument is considered as a program, and the second one is considered as an input. Then $\KS_U(x\cnd y)$ is defined as
\[
\KS_U(x\cnd y) = \min\{ |p|\colon U(p,y)=x\}.
\]
Again there is an optimal $U$ that makes the function $\KS_U$ minimal (as a function of two arguments) up to $O(1)$. See the details, e.g., in~\cite{suv}. For simplicity we discuss the non-conditional setting when possible.

\section{Kolmogorov optimality or Solomonoff universality?}

Kolmogorov was not the first one who considered the notion of optimal encoding: Ray Solomonoff considered this notion earlier in the framework of inductive inference theory. He published some technical reports in the beginning of 1960s, and in 1964 he published a paper~\cite{solomonoff1964-1} where essentially the same optimality result was proven. For that, Solomonoff introduced the notion of a \emph{universal} machine: a machine $U$ is universal if for every machine $V$ there exists a string $v$ such that $U(vx)\simeq V(x)$ for all strings $x$, i.e., $U(vx)$ and $V(x)$ are defined for the same strings $x$ and are equal if defined. In other words, universal algorithm $U$ may simulate any other algorithm, one needs just to add some prefix. Solomonoff explains it as follows:

\begin{quote}
More exactly, suppose $M_2$ is an arbitrary computing machine, and $M_2~(x)$ is the output of $M_2$, for input string $x$. Then if $M_1$ is a ``universal machine," there exists some string, $\alpha$ (which is a function of $M_1$ and $M_2$, but not of $x$), such that for any string, $x$,
\[
M_1(\overline{\alpha x}) = M_2(x)
\]
$\alpha$ may be viewed as the "translation instructions" from $M_1$ to $M_2$. Here the notation $\overline{\alpha x}$ indicates the concatenation of string $\alpha$ and string $x$. \cite[p.~7]{solomonoff1964-1} 
\end{quote}
It is easy to see that every Solomonoff universal algorithm is optimal in the sense of Kolmogorov. However, the reverse is not true: one can easily construct a Kolmogorov optimal algorithm that is defined only on strings of even length (just prepend the program by $1$ or $01$ to make its length even), and for a Solomonoff-universal algorithm this is not possible, since to simulate a total $V$ the universal algorithm should be defined on all extensions of $v$.  

Therefore, if we decide to choose some $U$ arbitrarily and then fix it and let $\KS(x) = \KS_U(x)$, there are still two options: to choose any (Kolmogorov) optimal $U$ or to insist that $U$ should be not only optimal but also (Solomonoff) universal.  It is interesting to see how different authors make this choice. Chaitin in his first paper~\cite{chaitin1966} does not consider optimality at all (and does not define the notion of algorithmic complexity as it is known now); instead, he considers some parameters of a specific computation model (Turing machines). In his second paper~\cite{chaitin1969} he gives a specific construction of an optimal machine $M^*$ based on a Turing universal machine and  a specific encoding of programs, and notes that this machine is optimal:
\begin{quote}
(c)~For any binary computer $M$ there exists a constant $c$ such that for all finite binary sequences $S$, $L_{M^*}(S)\le L_M(S)+c$. \cite[p.~156]{chaitin1969}
\end{quote}
(Chaitin uses letter $L$ instead of our $\KS$.)

Li and Vit\'anyi in their famous textbook~\cite{li-vitanyi2019} define complexity in terms of optimality. Lemma 2.1.1~\cite[p.104]{li-vitanyi2019} says ``There is an additively optimal universal partial computable function'', but the word ``universal'' has no precise definition in the book and refers to the construction of this optimal function. In Definition 2.1.2 (p.~106) the authors say ``Fix an additively optimal universal $\phi_0$'' (again without defining universality in a precise way). One can also note that in the 2nd edition (1997) the word ``universal'' was formally introduced as a synonym for ``optimal'' in Definition 2.0.1 (p.~95) : ``A function $f$ is \emph{universal} (or \emph{additively optimal}) if\ldots''.

Downey and Hirschfeldt~\cite{downey-hirschfeldt2010}, on the other hand, define Kolmogorov complexity using universal algorithms. They write:
\begin{quotation}
We would like to get rid of the dependence on $f$ by choosing a \emph{universal} description system. Such a system should be able to simulate any other description system with at most a small increase in the length of descriptions. Thus we fix a partial computable function $U \colon 2^{<\omega} \to 2^{<\omega}$ that is \emph{universal} in the sense that, for each partial computable function $f \colon 2^{<\omega} \to 2^{<\omega}$, there is a string $\rho_f$ such that
\[
\forall \sigma [U(\rho_f\sigma) = f(\sigma)].
\]

[\textsc{Footnote}:] Such a function, and the corresponding Turing machine, are sometimes called \emph{universal by adjunction}, to distinguish them from a more general notion of universality, where we say that $U$ is universal if for every partial computable $f \colon 2^{<\omega} \to 2^{<\omega}$ there is a $c$ such that for each $\sigma\in\mathrm{dom}(f)$ there is a $\tau$ with $|\tau|\le |\sigma|+c$ and $U(\tau)=f(\sigma)$. The distinction between these two notions is irrelevant in most cases, but not always, and it should be kept in mind that it is the stronger one that we adopt.
\end{quotation}
So they do make the distinction between optimality (called ``a more general notion of universality'') and universality and choose the stronger requirement (but, as it seems, do not say explicitly when this distinction really matters).

Andre Nies in his textbook~\cite{nies2009} defined optimality (Definition 2.1.1, p.~76) in a standard way, and then writes:
\begin{quote}
An optimal machine exists since there is an effective listing of all the machines. The particular choice of an optimal machine is usually irrelevant; most of the inequalities in the theory of string complexity only hold up to an additive constant, and if $R$ and $S$ are optimal machines then $\forall x\, C_R(x)=^+ C_S(x)$. Nonetheless, we will specify a particularly well-behaved optimal machine. Recall the effective listing $(\Phi_e)_{e\in\mathbb{N}}$ of partial computable functions from (1.1) on page 3. We will from now on assume that $\Phi_1$ is the identity.

\textbf{2.1.2 Definition}. The \emph{standard optimal plain machine} $\mathbb{V}$ is given by letting 
\[
\mathbb{V}(0^{e-1}1\rho)\simeq \Phi_e(\rho)
\]
for each $e>0$ and $\rho\in\{0,1\}^*$.
\end{quote}
Still it is not clear whether indeed some of the (non-technical) results in the book depend on this stronger assumption and are not easily extended to all optimal machines.

\section{How invariant are the theorems?}

As we have seen, different authors use technically different definitions for Kolmogorov complexity. Is it important? Are there some results that are valid for one definition but not for the other?

First let us stress again (see Nies' quote above) that many results about Kolmogorov complexity are \emph{$O(1)$-robust} in the sense that changing the complexity function by $O(1)$ additive terms does not change anything. For example, when we claim that
\[
\KS(x,y) \le \KS(x)+\KS(y)+O(\log n)
\]
for all strings $x,y$ of length at most $n$, the statement remains valid if we change $\KS(x)$ by some bounded additive term. Indeed, the $O(\log n)$ term absorbs this change. For $O(1)$-robust results it obviously does not matter whether we consider optimal or universal algorithms in the definition of complexity (and which optimal or universal algorithm we choose).

Some standard results are not $O(1)$-robust in their standard form but the proof essentially establishes some stronger $O(1)$-robust statement. For example, one of the first results about Kolmogorov complexity says that Kolmogorov complexity function is non-computable, and computability is not an $O(1)$-robust property.  However, the proof establishes that every function that $O(1)$-approximates complexity is non-computable (and, moreover, that every computable partial lower bound for complexity is bounded) --- and these stronger statements are $O(1)$-robust.

In some cases the statement is not $O(1)$-robust and there is no stronger $O(1)$-robust statement. For example, for every $n$ there exists an $n$-bit string such that $\KS(x)\ge n$, and this is not true for function $\KS(x)-c$ for large enough $c$. In this example the statement is valid for every $\KS_U$ (for every optimal~$U$, and even for a non-optimal $U$). Also one may note that we often do not need such a precision when we apply this statement. Usually it is enough to know that for some $c$ and all $n$ there is a string of length $n$ that has complexity at least $n-c$, and this statement is $O(1)$-robust.

There are some cases where the choice of a universal (or optimal) machine is essential and therefore there is no $O(1)$-robust replacement. There is a very convincing example for prefix complexity: An.~Muchnik and S.~Positselsky~\cite[p.~29, Theorems 2.6 and 2.7]{muchnik-positselsky2002} have shown that the overgraph of the prefix complexity is $tt$-complete for some prefix universal functions and is not $tt$-complete for some others\footnote{Prefix complexity is obtained when we require additionally that the domain of the algorithm~$U$ is prefix-free (does not contain at the same time a string and its prefix). See below or~\cite{suv} for details.}. For the standard (plain) version there are trivial examples of this type: one may easily show that for some universal $U$ the function $\KS_U$ takes only even values, while for other universal $U$ the function $\KS_U$ takes only odd values.

\begin{question}
Are there more interesting examples of this type (for plain complexity)?
\end{question}

\section{Plain complexity: universality vs optimality}

The following proposition shows that for plain complexity all the results valid for $\KS_U$ with (arbitrary) universal $U$ remain valid for all optimal algorithms~$U$. Let us recall the definitions.

\begin{definition}
  A machine $U$ is {\em optimal} if for every other machine $V$ there exits a constant~$c_V$ such that $\KS_U(x) \le \KS_V(x) + c_V$ for all $x$.
\end{definition}

\begin{definition}
  A machine $U$ is {\em universal} if for every other machine $V$ there exits a string $v$ such that $U(vp)\simeq  V(p)$for all $p$. 
\end{definition}

(The sign $\simeq$ means that either both quantities are undefined or both are defined and equal.) 

It turns out that the classes of functions $\KS_U$ for all optimal $U$ and for all universal $U$ are the same, even if the class of optimal algorithms is bigger. 

\begin{proposition}\label{optimal2universal}
  Let $D$ be an optimal algorithm. Then there exists some universal algorithm $U$ such that $\KS_U(x)=\KS_D(x)$ for all strings $x$.
\end{proposition}

\begin{proof}
Let us prove first that \emph{for some $d$, none of the strings $0^dx$} (i.e., strings starting with $d$ zeros) \emph{appear as shortest $D$-descriptions}. Indeed, the complexity of a string $0^d x$ of length $n$ does not exceed $n-d+2\log d + c$ for some $c$ that does not depend on $d$ and $x$ (start with self-delimited description of $d$ of length $2\log d +c$ and then add $x$ of length $n-d$). If such a string is a $D$-description of some $y$, then the complexity of $y$ does not exceed $n-d+2\log d+c'$ where again the constant $c'$ does not depend on $d,x,y$.  So for large enough constant $d$ the description $0^d x$ will never be a shortest $D$-description of any string $y$, since $n-d+2\log d + c' <n$ for all $n$.

Fix some $d$ with this property. Then the complexity function $\KS_D$ does not increase if we redefine $D$ on strings starting with $d$ zeros and let $D(0^d z)$ be equal to $U(z)$ where $U$ is some universal algorithm. Indeed, all the shortest $D$-descriptions remain intact. On the other hand, this change makes $D$ universal.

There is still a possibility that the complexity could \emph{decrease} after this change, but this will not happen for large enough $d$ (since $D$ was optimal, $U$ provides descriptions that are at most $O(1)$-shorter compared to $D$-descriptions).  
\end{proof}

Let us note also that some theorems in the textbooks about Kolmogorov complexity use not only the complexity function but also refer to the decompressor explicitly. For example, one can show that \emph{for any optimal $D$ there exist some constant $c$ such that there exist at most $c$ shortest $D$-descriptions of $x$}. This statement cannot be formulated just in terms of the complexity function, so for every optimal (or universal) $D$ it needs to be proven anew, and this is not difficult, see \cite[p.~40, Ex.~40, \emph{Hint}]{suv}.  Note also that one can weaken the restriction and consider all descriptions that are at most $d$ bits longer than the optimal ones. Then $c$ depends on $d$, but the statement remains true. 

\begin{question}
What are other examples of statements of this type? Are there cases when we need to know that $D$ is universal (by adjunction), not only optimal?
\end{question}

An interesting example of this type was suggested by J.~Miller for the case of prefix complexity. Let $D$ be a prefix-free machine. Then, for each string $x$, we consider the probability $m_D(x)$ to get output $x$ on a random input, i.e., the sum
\[
m_D(x)= \sum \{2^{-|p|}\colon D(p)=x\}
\]
If $D$ is universal, all the reals $m_D(x)$ are random (for the same reason as for the Chaitin $\Omega$-numbers). On the other hand, using Kraft--Chaitin construction, we can get a prefix-free $D$ such that for every $x$ there exists exactly one program for $x$ for each of the lengths $\KP(x)+1$, $\KP(x)+2$, etc. For this $D$ the sum $m_D(x)$ is equal to $2^{-\KP(x)-1}$ (and is a power of $2$, not a random number).

The proof of Proposition~\ref{optimal2universal} can be modified for the case of conditional complexity (descriptions of the form $0^dx$ cannot be the shortest ones if $d$ is sufficiently large, and those descriptions can be reused to achieve universality). Therefore, for every optimal $D$ there exists a universal $U$ such that $\KS_U(x\cnd y)=\KS_D(x\cnd y)$ for all $x$ and $y$. %\comm{I just wanted to say something before people decide that conditional complexity is too much for them!}

\section{Prefix complexity: universality vs optimality} 

\subsection{Unconditional prefix complexity}
There is a version of complexity that was introduced by Levin and later independently by Chaitin, and now is usually called (unconditional) \emph{prefix complexity}. There are three different versions of the definition.

\emph{First one}: in the definition of $\KS_U$ we consider only prefix-free algorithms $U$. Here \emph{prefix-free} means that $U(p)$ and $U(p')$ cannot be both defined for some string $p$ and its extension~$p'$, i.e., a longer string $p'$ that has $p$ as a prefix.  In other terms: no two different strings in the domain of $U$ can be comparable in the sense that one is a prefix of the other one. If we restrict ourselves to prefix-free algorithms $U$, we again have an optimal algorithm \emph{among them}. There is also a prefix-free algorithm that is universal among prefix-free algorithms. This is how Chaitin introduced the prefix complexity.

\emph{Second one}: in the definition of $\KS_U$ we consider only prefix-stable algorithms. Here \emph{prefix-stable} means that if $U(p)$ is defined for some $p$, then $U(p')$ is defined and is the same for all $p'$ that are extensions of $p$. Again, there exist optimal and even universal algorithms among all prefix-stable algorithms. This definition goes back to Levin (but is less popular now, unfortunately).

The \emph{third definition} is more abstract and deals only with functions (but not with the algorithms): we consider upper semicomputable functions $K$ with integer values defined on strings such that $\sum_x 2^{-K(x)}\le 1$. Upper semicomputability means that $K(x)$ is a limit of a uniformly computable   decreasing sequence of total computable functions $K_n(x)$. There exists a minimal function (up to $O(1)$ additive term) in this class, and it is called prefix complexity.

Prefix complexity can also be defined as the
negative logarithm of maximal lower semicomputable semimeasure. % Sounds too brief or too long \bruno{This is similar to $2^{-K(x)}$ in the third definition, but values are not negative powers of two.}  
We will not go into details here (see, e.g.,~\cite{suv}), but briefly use this characterization in Section~\ref{omega} 
about $\Omega$-numbers.
\medskip

Let us now return to three definitions listed above. It turns out that these three definitions are equivalent in a rather strong sense:

\begin{proposition}\label{prefix}
\leavevmode

\textup{(a)}~The classes of functions $\KS_U$ for all prefix-stable $U$ and for all prefix-free $U$ are the same and coincide with the class considered in the third definition.
 
 \textup{(b)}~For the first two definitions, the class of functions that correspond to universal prefix-free \textup[prefix-stable\textup] algorithms coincides with the class of functions that correspond to optimal prefix-free \textup[\emph{resp.} prefix stable\textup] algorithms, and therefore consists of $O(1)$-minimal functions in the class mentioned in~\textup{(a)}.
\end{proposition}

\begin{proof}
\textup{(a)}~It is easy to see that all functions $\KS_U$ for prefix-free or prefix-stable $U$ are upper semicomputable and  $\sum_x 2^{-\KS_U(x)}\le 1$. It remains to show the inclusion in the other direction. Let $K(u)$ be some upper semicomputable function on strings with integer values such that $\sum 2^{-K(x)}\le 1$. We need to construct a prefix-free algorithm $U$ such that $\KS_U(x)=K(x)$ for all $x$, and also a prefix-stable $U$ with the same property.

There are two different cases. The first case happens if $\sum 2^{-K(x)}$ is exactly $1$. Then the function $K$ is computable. Indeed, to check whether the current value of $K$ is final and will not decrease again, we wait until the sum of $2^{-K(x)}$ becomes close to $1$ enough to prevent the decrease (this means that decreasing $K(x)$ will make this sum greater than $1$). It remains to use the computable Kraft--Chaitin lemma (see, e.g.,~\cite[p.94]{suv}) to assign a description of length $K(x)$ to each $x$. (This can be easily done both in a prefix-free and prefix-stable setting.)

  Now assume that the sum of the series $\sum_x2^{-K(x)}$ is strictly less than $1$. Then one can find a tail of this series that can be doubled and still the sum is at most $1$. Using a non-uniform construction, we may assume that we know where this tails starts and the values of $K(x)$ for $x$ that are not in this tail. Then we again apply the Kraft--Chaitin lemma to the exact values of $K(x)$ for $x$ not in the tail, and for all pairs $(K(x),m)$ such that $x$ is in the tail and $K(x)\le m$ (in the order of their enumeration). The total sum of all values of $2^{-m}$ does not exceed the doubled sum of the tail, so the conditions of the Kraft--Chaitin lemma are satisfied.

\textup{(b)}
For the prefix-free algorithms, we need to take some program that is not optimal (such a program always exists, otherwise the prefix complexity function would be computable), and then repurpose it to achieve the universality requirement (with additional zeros, to avoid creating shorter programs).

  For the prefix-stable algorithms the situation is different. The descriptions for some $x$ form one or several cones rooted at  prefix-minimal descriptions (whose prefixes are not descriptions). If for some $x$ we have more than one cone, we can repurpose one of them (where the root is longer or has the same length) without losing optimal descriptions; if needed, we extend $x$ farther if avoid shorter descriptions. The choice of the prefix is hard-coded in the new decompressor. So it remains to prove the following lemma:
\begin{lemma}
Let $U$ be an optimal prefix-stable algorithm. Then there exists some string $x$ that has two incomparable descriptions.
\end{lemma}

  \begin{proof}[Proof of the lemma]   Assume that no string has incomparable descriptions (for some prefix-stable~$U$). This means that the descriptions of each string form only one cone: there exists a shortest description, and all other descriptions are extensions of it. We want to prove that $U$ cannot be optimal. Simulating $U$ on all inputs, we get, for each string $x$, some description of~$x$. It corresponds to some aligned interval in $[0,1]$ (binary fractions that have this description as a prefix). When a shorter description of the same string is found, and it is the prefix of the previously known description, this interval increases. And if an incomparable description of the same string appears, we know that some common prefix of these two descriptions will appear later, due to our assumption. So if we have, say, $2^n$ descriptions of $2^n$ different objects (strings), we have $2^n$ disjoint intervals in $[0,1]$ that may grow with time, and we can computably select one that has at most $O(2^{-n})$ % \comm{Why big O?} because I don't want to use explicit constant $2$ that seems necessary
space for growth, being bounded by the neighbor intervals. This guarantees that its maximal future size is $O(2^{-n})$ and therefore the complexity of the corresponding object is at least $n-O(1)$. Since this process is computable (given $n$), optimality would require the complexity of this object to be $O(\log n)$,  a contradiction.
\end{proof}
Using this lemma, we finish the proof of the proposition as explained above.
\end{proof}

\begin{question}
  Is there a similar (``algorithms-free'') description of the class of all functions $\KS_U$ for all (not necessarily prefix-free or prefix-stable) algorithms $U$? For every $U$, the function $\KS_U$ is lower semicomputable and there are at most $2^n$ objects $x$ such that $\KS_U(x)=n$, but are these properties strong enough to be a characterization of the class $\KS_U$ for all $U$? % Probably not, but why?
\end{question}

\subsection{Prefix conditional case}

The arguments used to prove the results of the previous section are not directly applicable for conditional prefix complexity $\KP(x\cnd y)$ since the reasoning is not uniform (and now we need a construction that computably depends on $y$). To discuss the situation in more details, let us recall that ``plain'' (non-prefix) conditional complexity  $\KS(x\cnd y)$ is defined as follows. We let
\(
\KS_D(y\cnd x)= \min\{ |p|\colon D(p,y)=x\}
\)
where $D$ is some algorithm applied to pairs of strings. 
Algorithm $U$ is optimal if function $\KS_U$ is minimal up to $O(1)$ in the class of all functions $\KS_D$ for all algorithms $D$. Algorithm $U$ is universal if for every algorithm $V$ there exists a string $v$ such that
\[
V(p,y)\simeq U(vp,y)
\]
for all $p$ and $y$. As before (for the unconditional case), universality implies optimality.

For the prefix case (and ``plain'' conditions) we consider prefix-free or prefix-stable algorithms. Algorithm $D(p,y)$ is \emph{prefix-free} if $D(p,y)$ and $D(p',y)$ are never defined both for some string $p$ and its extension $p'$ \emph{and the same\footnote{This condition is often obscured by informal language of ``self-delimited programs'', which sometimes causes misunderstanding, since the intuition of a ``self-delimited program'' assumes implicitly that the program is self-delimited by itself, without $y$. But in reality a prefix-free algorithm decides when to stop reading the program tape, \emph{depending on the content of the condition tape}.} $y$}. Algorithm $D(p,y)$ is \emph{prefix-stable} if $D(p,y)=x$ implies $D(p',y)=x$ for every $p'$ that is an extension of $p$.  As for the plain complexity, there exist universal prefix-stable and prefix-free algorithms, and they are optimal in their classes. The corresponding (conditional) complexity functions differ at most by $O(1)$ additive term. Another way to define conditional prefix complexity is to consider a minimal function in the class of upper semicomputable integer functions $K(x\cnd y)$ of two string arguments such that $\sum_x 2^{-K(x\cnd y)}\le 1$ for all $y$. All three definitions of conditional complexity (using prefix-free and prefix stable algorithms, or the class of functions), are equivalent up to $O(1)$-precision. However, the situation here is more complicated than in the unconditional case. 

\begin{question}
  Assume that we have an integer-valued function $K(x\cnd y)$ on pairs of strings such that $\sum_x 2^{-K(x\cnd y)}\le 1$ for all $y$. Can we find a prefix-free (or prefix-stable) conditional algorithm $U$ such that $\KS_U(x\cnd y)=K(x\cnd y)$?
\end{question}

\begin{question}
If the answer is negative, can we show that at least prefix-free and prefix-stable function give the same class? Assume that $U$ is a prefix-stable function; is there a prefix-free function $U'$ such that $\KS_U=\KS_U'$? (In the other direction it is true for trivial reason: we may extend any prefix-free function above.)
\end{question}

The argument for the non-conditional case (Proposition~\ref{prefix},~(a)) cannot be used directly since it is non-uniform (for each $y$ we need some advice, and this is not finite information in total). The argument for (b)  is non-uniform either, since the non-optimal description (for a given condition $y$) may have unbounded length. Still for prefix-free machines it can be applied after some additional preparations.

\begin{proposition}
Let $U(p,y)$ be an optimal prefix-free algorithm for conditional complexity. Then there exists a universal prefix-free algorithm $U'(p,y)$ such that $\KS_{U'}(x\cnd y) = \KS_U(x\cnd y)$.
\end{proposition}

\begin{proof}
To use the same argument as in the plain case, we need to modify our optimal prefix-free algorithm in such a way that it is not defined on strings that start with $d$ zeros for some fixed $d$, the same for all conditions. Then this empty space can be used to embed some universal algorithm in the same way as before. This can be done in two steps. First, we made the following observation:
\begin{lemma}
For every optimal prefix-free $U$ there exists some $\varepsilon>0$ such that for every $y$ the sum $\sum_x2^{-\KS_U(x\cnd y)}$ does not exceed $1-\varepsilon$.
\end{lemma}

  Note that for every $y$ the sum is less than $1$, otherwise $\KS_U(x\cnd y)$ as a function of $x$ would be computable (and it differs from unconditional complexity by at most $O(1)$). The lemma says that the sums can be separated from $1$ uniformly in $y$.

\begin{proof}[Proof of the lemma]
Assume that the statement does not hold. Then for every~$n$ there exist some $y_n$ such that 
\[
\sum_x 2^{-\KS_U(x\cnd y_n)}>1-2^{-n}.
\]
Since the sum is uniformly lower semicomputable, we can find those $y_n$ effectively. 

Knowing $n$ and $y_n$, we can wait until the sum in their definition crosses the threshold $1-2^{-n}$. At that moment we know all strings $x$ such that $\KS_U(x\cnd y_n)<n$ (if a new string of complexity less than $n$ appears, the decrease in complexity would make the sum greater than $1$), and can select some string $z_n$ that is not in the list. Then $\KS_U(z_n\cnd y_n)\ge n$ while $\KS(z_n)=O(\log n)$, since $z_n$ can be computed given $n$, a contradiction (conditional complexity could not be much bigger than unconditional one).
\end{proof}

Now, using Kraft--Chaitin lemma, we rearrange the descriptions in $U$: for every $y$ we enumerate all $p$ such that $U(p,y)$ is defined. Their lengths form a sequence that satisfies the conditions of the Kraft--Chaitin lemma even after some large $d$ is prepended to it. Applying the lemma, we get some string of length $d$ that is free for inserting there a universal algorithm (and does not depend on $y$ since it is generated by Kraft--Chaitin allocator when only $d$ was given to it). This string then can be used in the same way as in the plain case.
\end{proof}

Still for prefix-stable case this argument does not work. The problem is that in this case a prefix of a description can later appear as a description (of the same object).

\begin{question}
Let $U(p,y)$ be an optimal prefix-stable algorithm. We look for a universal prefix-stable algorithm $U'$ such that $\KS_U(x\cnd y)=\KS_{U'}(x\cnd y)$ for all $x,y$. Can such an algorithm always be found?
\end{question}

\section{Chatin's $\Omega$-number[s]}\label{omega}

Chaitin defined his famous $\Omega$-number as the probability $\Omega_U$ of termination of a universal prefix-free algorithm $U$ on a random sequence: we apply $U$ to random bits until it terminates on some prefix of this random bit sequence. 

In~\cite{chaitin1977} he suggests some specific universal $U$ based on LISP, but it is more natural to consider this definition as a definition of a \emph{class} of reals (all reals that can be obtained in this way using different universal prefix-free $U$). They can be called \emph{Chaitin $\Omega$-numbers} (or just $\Omega$-numbers if the context is clear).
\medskip

Several equivalent characterizations were found for $\Omega$-numbers:
\begin{itemize}
\item Martin-L\"of random lower semicomputable reals in $(0,1)$;
\item lower semicomputable reals in $(0,1)$ that are maximal with respect to Solovay reducibility;
\item sums $\sum_n m(n)$ for maximal lower semicomputable semimeasures on $\mathbb{N}$.
\end{itemize}

\noindent
See~\cite[Section 5.7]{suv} for more details, historical references, and equivalence proofs.\footnote{The argument in~\cite[Theorem 103]{suv} proves only that a Solovay maximal lower semicomputable real is a sum of a maximal semimeasure. But in fact similar reasoning shows that every Solovay maximal semicomputable real $\omega$  is a termination probability for some universal prefix-free algorithm. We note that $\omega = 2^{-k}\Omega_U+\alpha$ for fixed prefix-free universal $U$, some $k$ and some lower semicomputable $\alpha$, and modify $U$ by adding prefix $0^k$ to all programs and using the rest of the Cantor space to get additional termination probability $\alpha$. Note that since $\omega<1$, for large $k$ there is enough space, since $\alpha\le\omega\le 1-2^{-k}$ for large enough $k$.}
\medskip

These results imply that for every prefix-free optimal $U$ (and therefore for every prefix-free universal~$U$) the sum $\sum_x2^{-\KS_U(x)}$ is an $\Omega$-real, since $2^{-\KS_U(x)}$ is a maximal lower computable semimeasure on $\mathbb{N}$. The reverse is also true, so we get one more characterization of $\Omega$-numbers:

\begin{proposition}
For every $\Omega$-number $\omega$ there exists a prefix-free universal $U$ such that $\Omega = \sum_x2^{-\KS_U(x)}$.
\end{proposition}

\begin{proof}
According to Proposition~\ref{prefix}, it is enough to construct an upper semicomputable $K$  such that
\begin{itemize}
\item $\sum_x 2^{-K(x)}\le 1$;
\item $K$ is minimal in this class;
\item $\sum_x 2^{-K(x)} = \omega$. 
\end{itemize}
Then Proposition~\ref{prefix}~(a) converts this $K$ into an optimal $U$ with required properties, and then Proposition~\ref{prefix}~(b) converts it into a universal $U$ .

To achieve this goal, let us start with some minimal function $K$ in this class. Since $\omega$ is Solovay maximal, we can find $d$ such that $\omega - 2^{-d}\sum_x 2^{-K(x)}$ is lower semicomputable. We may increase $K$ by $d$ and assume without loss of generality that
\[
\sum_x 2^{-K(x)}+\alpha = \omega
\]
for some lower semicomputable $\alpha$. What we need to do is to decrease $K$ and obtain an upper semicomputable $K'\le K$ such that 
\[
\sum_x 2^{-K'(x)} = \sum_x 2^{-K(x)}+\alpha.
\]
This can be done for any upper semicomputable $K$ and lower semicomputable $\alpha$.

To achieve this, we consider approximations $K_t$ for $K$ from above (they decrease as $t$ increases), and rational approximations $\alpha^t$ for $\alpha$ from below (they increase as $t$ increases). At every step $t$ we construct a computable approximation $K'_t$ for $K'$ that decreases as $t$ increases. We may assume that $K_t(x)$ is infinite for all $x$ except for finitely many of them (and the list of those $x$ is a part of finite object $K_t$).

We let $K'_0=K_0$ and then maintain that
\begin{itemize}
\item $K'_t \le K_t$ everywhere;
\item $\sum_x 2^{-K_t(x)} +\alpha_n-1/n \le \sum_x 2^{-K'_t(x)} \le \sum_x 2^{-K_t(x)}+\alpha_n$.
\end{itemize}
As $t$ increases, $K_t$ decreases and may become less than $K'_t$ thus violating the first condition. We have to decrease some values of $K'_t$ to maintain this condition. This (if done in a minimal way) will not violate the second inequality in the second condition since the increase in $2^{-K'_t(x)}$ does not exceed the increase in $2^{-K_t(x)}$, and $\alpha_n$ may only increase. It remains to satisfy the first inequality in the second condition (not violating the second inequality). This is always possible since $K'_t$ has arbitrary large values and by decreasing them we can increase the sum $\sum_x 2^{-K'(x)}$ by steps that are as small as needed. 

As $t\to\infty$, we get a limit function $K'$ that is upper semicomputable, does not exceed $K$ and $\sum_x 2^{-K'(x)}=\sum_x 2^{-K(x)}+\alpha=\omega$, as required. This finishes the proof.
\end{proof}

\section{Specific universal algorithms}

All the questions discussed above would disappear if everyone could agree on some specific programming language (decompressor) $D$. There are some other reasons why this sounds attractive. First of all, after that the expression ``$x$ is simpler than $y$'' would have some objective (even if somehow arbitrary) meaning. Also we would be able to use some specific constants in the theorems of algorithmic information theory. 

Here is how Levin explains his motivation for suggesting some specific optimal measure of complexity (and corresponding algorithm):
\begin{quotation}
In concrete considerations involving the algorithmic approach to information theory (which is based on the above-mentioned [optimality] theorem and a number of others like it), it is important to choose an optimal measure of complexity such that the constant is not too large when the measure is compared with other natural measures. In this regard, Kolmogorov's remark (see [1]) [cited above, number~\cite{kolmogorov1965} in our bibliography] is of interest; namely that in his opinion it is possible to see to it that the constant in question does not exceed a few hundred bits. However, no indication is given in [1] of how the complexity measure should be chosen in order to achieve this. It should be also noted that in a number of constructions the complexity appears in the exponent of important expressions and that in these cases a practically acceptable size of the constant can be measured only in tens of bits. Of course, we would ideally like the arbitrariness in the choice of the optimal complexity measure among natural ones to be measured in only a few bits, but to achieve this we must have a very explicit idea of the concept of ``naturalness'', which is hard to accomplish with such narrow limits on the arbitrariness.

Below, in \S 4, we will give a method allowing us to construct concrete optimal functions for defining complexity measures and a wide class of other analogous measures. It seems likely that the constants that appear in comparing the functions given below with the functions defined by other natural methods will not exceed a few hundred bits. Our construction allows several variations and possibly some refinements not noted here. We think that in one of these versions the size of the constant does not even exceed several tens of bits.

The constructions we use are very simple, and consequently, it seems, not original in most cases. The usefulness of this paper may reside in the selection of the goals being considered, rather than in its inventiveness. \cite[p.~727]{levin1977}
\end{quotation}
Still it seems that Levin's universal algorithm (or corresponding programming language) was never used later for any practical or theoretical purpose.

Some other specific universal algorithms were suggested by others. Chaitin in his book~\cite{chaitin1977} defines some version of LISP, explicitly construct several LISP programs, and defines complexity in terms of some specific universal computer working with LISP programs~\cite[Section 6.2]{chaitin1977}. Unlike Levin, he does not make any claims about conciseness of this programming language, and, like for Levin's approach, it seems that this specific universal algorithm was never essentially used later for any practical or theoretical purpose by others. Chaitin used it to construct an explicit description for a specific random $\Omega$-number in terms of exponentially Diophantine equations, but it is more a \emph{tour de force} than some practical application.

John Tromp~\cite{tromp2007} switched from LISP to combinatory logic where the universal combinator can be compressed to only $425$ bits, and computed some explicit constants. For example, he proved that for his universal algorithm (its prefix version) one has
\[
\KP(x,y) \le \KP(x)+\KP(y\cnd x^*)+1876
\]
(see~\cite[p.~255]{li-vitanyi2019} for more details).

From a programmer's view, the difference between complexities with respect to different programming languages can be described as follows: we ask what is the minimal length of an interpreter of a programming language $A$ written in programming language $B$. Now, when even the simplest application usually occupies several megabytes, only old people remember that in the era of $8$-bit computers with 32Kbytes memory an interpreter (or compiler) of a not so simple language (like \texttt{Pascal}) could sometimes fit into several kilobytes. This is only a rough analogy, but this suggests the sizes of the constants in the range $10^4\ldots10^5$ (recall that byte is $8$ bits).

Let us also mention another approach to provide some specific numbers as ``complexities'' of strings that was suggested in~\cite{delahaye-zenil2007}. Recall that asymptotically the complexity can be identified with a priori probability (technically, prefix complexity differs from discrete a priori probability at most by $O(1)$, see, e.g.,~\cite{suv}). To find some ``approximation'' of the a priori probability, one can consider some natural class of computing devices (e.g., all Turing machines with given alphabet and given number of states), take a random one uniformly from the class and run it on a random input (or on an empty input). The empirical probability of obtaining a given string $x$ in such an experiment can the be considered as a replacement for a (theoretical) a priori probability.

\section{Proving non-randomness}

One of the possible practical application where the numeric values of Kolmogorov complexity are important is \emph{proving non-randomness in real life}, or more technically, rejecting a null hypothesis convincingly. This kind of application is discussed in~\cite{gurevich-passmore2012}. 

Consider an experiment where $1000$ bits are generated, and the null hypothesis is that they are generated by tossing a fair coin. In other words, the assumed distribution is the uniform distribution on $\mathbb{B}^{1000}$. Seeing thousand of zeros (tails) as outcome, any reasonable observer would reject the null hypothesis. Why? Because this event, all-zeros sequence $Z$, has extremely small probability $2^{-1000}$. But the same can be said for any other sequence $R\in \mathbb{B}^{1000}$ instead of the zero sequence $Z$: the probability of getting $R$ is the same $2^{-1000}$. So why we reject the null hypothesis seeing $Z$ but still accept it seeing $R$? 

This paradox has no answer in the classical probability theory. The answer suggested by algorithmic information theory: the difference between $Z$ and most other sequences $R$ is that $Z$ has low complexity (has a short description). More formally, we reject null hypothesis because not just a low probability event happens, but because a \emph{simple} low probability event happens.

What is the trade-off between the probability and complexity? What is a more convincing argument against a null hypothesis: an event of complexity $100$ and probability $2^{-1000}$ or an event of complexity $1000$ and probability $2^{-2000}$? Statisticians speak about the \emph{Bonferroni correction}: if we have a natural family of $k$ events with small probability $p$ (under the null hypothesis) and one of them (not specified in advance) happens, this is $k$ times less strong argument compared to one natural event with the same probability $p$. Translating this remark to algorithmic information theory, we may say that an event of complexity $100$ is an element of a family of about $2^{100}$ events that have the same or smaller complexity, so we should multiply the probability by $2^{100}$ (and get $2^{-900}$ in the first case compared to $2^{1000}\times 2^{-2000}=2^{-1000}$ for the second case, which therefore sounds more convincing). In general, an event $A$ of complexity $\KS(A)$ and probability $\Pr[A]$ should be considered as an argument against the null hypothesis with effective $p$-value
\[
2^{\KS(A)}\cdot \Pr[A]\eqno(*)
\]
According to this rule, seeing a sequence $x$ with length $n$ and complexity $n-d$ is an argument against fair coin hypothesis with $p$-value $2^{-d}$. In other word, the difference between length and complexity can be considered as ``randomness deficiency'' (with respect to fair coin distribution).

Recall the words of Levin who said that complexity is ``in the exponent of important expressions'' (and so an acceptable divergence between different complexity functions is only a few dozens bits); the expression for the effective $p$-value is one of them. It seems that there is no practical hope to use any specific complexity function for some formal procedure of establishing a standard for ``statistical proofs'' in the court, since such a precision seems completely unrealistic (cf. the examples discussed above).

Gurevich and Passmore suggested in~\cite{gurevich-passmore2012} to use a restricted language to specify ``focal events'' whose happening can be a statistical reason to reject the null hypothesis. May be now the large language models can be also used since they somehow provide the a priori probability (which roughly corresponds to $2^{-\KS(x)}$). They also somehow use the previous experience of the mankind as an oracle (or condition) since they are trained on a huge amount of data produced by Nature and humans (who are part of Nature), and this kind of oracle is desirable: for example, if the focal event was specified in advance, the existence of such a prophecy should decrease its complexity in $(*)$. Of course, using this approach we need to agree on which of many competing LLMs we use, which training data should be used, and so on.

\section{Combinatorial applications: examples}

Leaving science fiction and returning to mathematical questions, let us consider one specific situation where small constants in the theorems about Kolmogorov complexity are essential. An important type of the Kolmogorov complexity applications is proving combinatorial results using Kolmogorov complexity as a tool. A lot of examples can be found in the textbook of Li and Vit\'anyi~\cite[Chapter 6, The Incompressibility Method]{li-vitanyi2019}. They start with three examples (Section~6.1):
\begin{itemize}
\item Recognition of palindroms with a one-tape Turing machine requires $\Omega(n^2)$ steps (Lemma~6.1.1);
\item No-carry adder algorithm for pairs of $n$-bit strings (carry bits are collected separately, and then the carry bits vector is recursively added to the \texttt{xor} bits vector) uses no more than $1+\log n$ steps in average (Lemma 6.1.2);
\item There is a $n\times n$ Boolean matrix such that all its reasonably sized submatrices  have high rank (Lemma 6.1.3; it gives specific bounds for sizes and ranks).
\end{itemize}
The core idea is the same in all three examples: undesirable behavior of an object (say, a large number of steps in the second example, or a small rank submatrix in the third example) happens only when the object is compressible.  In other words, we assume that object has some undesirable property (e.g., a pair of strings requires a lot of operations during the addition) and conclude that this object has small complexity. Then we conclude that an incompressible object does not have this undesirable property, and, moreover, most objects do not have this property.

However, a careful reader will notice a problem in these arguments. For example, in the proof of Lemma 6.1.2 the authors note that if the addition algorithm requires $t$ steps, then strings $x$ and $y$ contain opposite bits in $t-1$ consecutive positions (that make carry propagation $t$ steps long), and in this case $x$ can be simplified given $y$ is a condition. Indeed, instead of specifying the bits in the opposite part of $x$, one can just specify the coordinates of this part (that requires only $O(\log n)$ bits), since the opposite bits in $y$ are included in the condition. Li and Vit\'anyi present this conclusion on p.~453 as $\KS(x\cnd n,y,q)\le n-t-1+\log n$ where (they say) $q$ is a program of $O(1)$ bits needed to reconstruct $x$. Note that there is no constant here, and this is not true if understood literally (and $q$ does not help, since we still need to use the optimality for conditional complexity, and it will bring some constant anyway). Still they need this bound without a constant, since the statement of Lemma~6.1.2 does not include constants.

This is not an isolated problem; it happens all the time when we try to derive some combinatorial results using Kolmogorov complexity. In Section~6.3 (Combinatorics) the authors explain:
\begin{quote}
The general pattern is as follows. When we want to prove a certain property of a group of objects (such as graphs), we first fix and incompressible instance of the object\ldots\ It is always a matter of assumed regularity in this instance to compress the object to reach a contradiction.
\end{quote}

The first example in Section~6.3 deals with tournaments (complete directed graphs: every two players played one game and one of them won). Theorem~6.3.1 (p.~460) uses incompressibility method to show that there is a tournament with $n$ vertices that does not have a transitive subtournament (a subset where ordering is linear) of size greater than $1+2 \lceil \log n\rceil$. Indeed, if there is a transitive subtournament of size $v$, then one can specify the tournament by listing all the vertices of the subtournament in increasing order (this requires $v \lceil \log n \rceil$ bits) and then specifying the remaining bits in some fixed order. In this way we replace $v(v-1)/2$ by $v\lceil\log n\rceil$ bits, and if $\lceil \log n\rceil < (v-1)/2$, the tournament is compressible.

The authors present this argument saying that $\KS(E(T)\cnd n,p)\le l(E'(T))$ where $l(E'(T))$ is the length of the description suggested above, $E(T)$ is an encoding of a tournament $T$ and ``$p$ is a fixed program that on input $n$ and $E'(T)$ \ldots outputs $E(T)$''. Again this is not true if understood literally, since the conversion to the universal algorithm for conditional complexity will introduce some constant.\footnote{There is another technical problem with the statement and proof of Theorem~6.3.1. The statement uses $\lfloor 2 \log n\rfloor$ instead of $2\lceil \log n\rceil$, and the proof says carelessly that ``it is easy to verify that $2\lfloor \log n \rfloor = \lfloor 2\log n\rfloor$ for all $n\ge 1$''. But this is a minor detail.}

This is not a problem for this proof: one should say that we get a bound for conditional complexity \emph{for some ad hoc decompressor} (which is not universal, by the way), and then say that for the corresponding complexity function there exists an incompressible string of length $n$. Essentially this argument is just a counting argument in disguise: we say that there are not enough bad objects since each of them has a short description (of some specific type).

\section{Combinatorial applications: discussion}

In general, we have seen that
\begin{itemize}
\item To prove a combinatorial statement of asymptotic nature, we may use Kolmogorov complexity, and the unspecified constants in $O(\cdot)$-terms that appear in Kolmogorov complexity results are usually absorbed by asymptotic terms in the statement.

\item If we want to prove a combinatorial statement without unspecified constants in the asymptotic terms, the constants in the statements about Kolmogorov complexity become a problem.

\item For some simple cases we can overcome this problem easily by saying that we use $\KS_U$ with respect to some specific algorithm $U$ in the proof. This often happens if the proof is essentially a counting argument in disguise.
\end{itemize}

However, there are arguments that use Kolmogorov complexity in a more involved way. For example, sometimes (as noted in~\cite[Exercise 6.1.1, p.~465]{li-vitanyi2019}) we use the symmetry of information property (the formula for the complexity of pairs), and this is possible only if we use an optimal algorithm in the definition of complexity.

So what can we do if we have a proof of some combinatorial results using Kolmogorov complexity and want do get rid of asymptotic terms with unspecified constants that appear due to $O(1)$-terms in complexity statements? There are several possibilities.

Sometimes we can introduce some additional parameter and consider the limit behavior. For example, the inequality
\[
2\KS(x,y,z)\le \KS(x,y)+\KS(x,z)+\KS(y,z)+O(\log n)
\]
for strings $x,y,z$ of length at most $n$ can be used to prove the following combinatorial statement: if $A$ is a subset of the product of three finite sets $X\times Y\times Z$,  and $A_{XY}$, $A_{YZ}$, and $A_{XZ}$ are its projections on $X\times Y$, $Y\times Z$, and $X\times Z$ respectively, then their cardinalities satisfy the inequality
\[
(\# A)^2 \le (\# A_{XY})\cdot (\# A_{YZ})\cdot (\#A_{XZ}).
\]
Note that this inequality does not involve any constants. To prove it, for an arbitrary $N$, we apply the inequality for complexities to an incompressible point in $A^N$ considered as a triple  made of $x\in X^N$, 
$y\in Y^N$, and $z\in Z^N$. As $N\to\infty$, the asymptotic term $O(\log N)$ becomes negligible compared to the linear terms that include logsizes of sets, and in the limit we get the required inequality  without any unspecified constants (see~\cite[p.12 and Chapter~10]{suv} for details).

However, in other examples this is more difficult. Here is one. Levin noted that there exist an ``everywhere complex sequence'', an infinite sequence of bits where every factor (substring) of length $n$ has complexity at least $0.99n - O(1)$. (Here $0.99$ can be replaced by arbitrary number smaller than $1$.) The original proof of Levin is reproduced in~\cite[Section 4.2]{durand-levin-shen2007} and uses (inductively) the complexity arguments to construct the sequence block by block.

The combinatorial counterpart of this statement goes as follows: assume that for every $n$ at most $2^{0.99n}$ strings of length~$n$ are chosen arbitrarily and declared ``forbidden''. Then there exist an infinite sequence $\alpha$ and some constant $c$ such that $\alpha$ does not contain forbidden substrings of length greater than $c$. This combinatorial statement can be derived from Levin's lemma (see~\cite[Section~8.5]{suv} for details), but this derivation does not give any specific value for $c$. Still we have a purely combinatorial statement, and it would be desirable to provide a specific bound. How can we achieve this?

\begin{itemize}
  \item We can fix some universal algorithm in the definition of Kolmogorov complexity and compute  specific constants in different theorems about Kolmogorov complexity.  Then the translation of Kolmogorov complexity argument will give some specific~$c$. This looks feasible, as the efforts in this direction (discussed above) show. But this will give some huge $c$ (recall that even after all the tricks to choose a concise programming language the constants turn out to be quite big).

  \item Another possibility is to translate the proof into another language that does not use Kolmogorov complexity explicitly, or just find some other proof. Here it is much more difficult than for the example above (where Kolmogorov complexity can be replaced by a counting argument in a straightforward way) but it is still possible. Several other proofs of the combinatorial statement were suggested (using Lovasz local lemma, or some kind of potential, or some advanced counting), and it turned out that a more general question was studied in algebra under the name of Golod--Shafarevich theorem. See~\cite[Section 8.5]{suv} and \cite{shen2018} for these proofs and historical references. There is even a ``tetris'' proof (we add random bits and cancel forbidden sequences) that uses Kolmogorov complexity, but in a different way that does not introduce unspecified constants.
\end{itemize}

\section{How to decrease constants?}

Still there are some cases when we do not know a proof of the combinatorial statement that does not use Kolmorogov complexity. An example of this type is aggregating opinions and a related combinatorial game described in~\cite{shen2024,kss2025}. We do not go into details here, but for this case the asymptotic bound obtained by a Kolmogorov complexity argument is better than the bound provided by standard tools of learning theory ($t^{1/2}$ instead of $t^{2/3}$).  Still the Kolmogorov complexity argument does not provide specific constants in the asymptotic bounds and does not provide a computationally efficient aggregating algorithm. The second problem seems unavoidable if we use Kolmogorov complexity, but one can indeed try to get some specific constants by fixing some universal (or optimal) algorithm and computing explicit constants in the theorems about complexity that are used in the proof.

There are two possible approaches. The first one is to use some specific programming language (e.g., combinatory logic language developed by Tromp, see above) and try to find explicit constants for statements about complexity used in the proof (there are a lot of them). This is more a programming challenge than a mathematical one.

Another possible approach worth trying is to start with some universal language and then try to modify it to get small constants in the statements used in the proof. For example, it may happen that we need the bound $\KS(A(x))\le \KS(x)+c$ for some specific algorithm $A$, all $x$ and some reasonably small constant $c$. Then we can decide to modify our initial optimal algorithm $U$ and get a new algorithm $U'$ as follows: $U'(1x)=U(x)$ and $U'(0x)=A(U'(x))$. Then the complexity function increases at most by $1$ (due to the first part of the definition) and $\KS_{U'}(A(x))\le \KS_{U'}(x)+1$ (due to the second part). One can perform similar modifications to get a bound for $\KS(xy)$ with small constant, etc. To get the reverse direction in the formula for pair complexity, one needs to use the fixed-point theorem (since the construction of the optimal algorithm needs to refer to itself).

Let us provide more details. In the following proposition we speak about conditional complexity $\KS_U(x\cnd y)$, and the unconditional complexity is defined by $\KS_U(x)=\KS_U(x\cnd \varepsilon)$ where $\varepsilon$ is an empty string. We fix some computable bijective pairing function $\langle x,y\rangle$. 
\begin{proposition}
There exist an optimal algorithm $U$ for conditional complexity such that the corresponding conditional and unconditional complexity functions satisfy the following inequalities:
\begin{itemize}
\item $\KS(x\cnd y)\le \KS(x)+10;$
\item $\KS(x\cnd x)\le 10;$
\item $\KS(f_1(x,y)\cnd y)\le \KS(x\cnd y)+10;$
\item $\KS(x\cnd y)\le \KS(x\cnd g_1(y))+10;$
\item $\KS(f_2(x,y)\cnd y)\le \KS(x\cnd y)+10;$
\item $\KS(x\cnd y)\le \KS(x\cnd g_2(y))+10;$
\item \ldots 
\item $\KS(\langle x,y\rangle\cnd z) \le \KS(x\cnd z)+2\log \KS(x\cnd z)+\KS(y\cnd \langle x,z\rangle)+20;$
\item $\KS(x\cnd z)+\KS(y\cnd\langle x,z\rangle)\le \KS(\langle x,y\rangle \cnd z)+6\log \KS(\langle x,y\rangle\cnd z)+30.$
\end{itemize}
\end{proposition}

The missing lines can contain, say, hundred different computable functions $f_1,\ldots$  and $g_1,\ldots$ (there is some slack in the constants, $10$ can be optimized) --- depending on what is needed for the applications. For example, one of this $f_i$ can be the function $\langle x,y\rangle\mapsto \langle y,x\rangle$, and then we have $\KS(\langle x,y\rangle)\le \KS(\langle y,x\rangle)+10$ (in the unconditional version, when the condition is $\varepsilon$).  Two last lines correspond to the formula for the complexity of pairs.

\begin{proof}
\leavevmode
 Let $\textbf{C}_0,\textbf{C}_1,\ldots$ be different string constants of length $10$. Let $U_0$ be some optimal algorithm for conditional complexity. Let
\begin{itemize}
\item $U(\textbf{C}_0p,y)=U_0(p,y);$
\item $U(\textbf{C}_1p,y)=U(p,\varepsilon);$
\item $U(\textbf{C}_2p,y)=y;$
\item $U(\textbf{C}_3p,y)=f_1(U(p,y),y);$
\item $U(\textbf{C}_4p,y)=U(p,g_1(y));$
\item \ldots
\end{itemize}
  The first line guarantees that $\KS_U(x\cnd y) \le \KS_{U_0}(x\cnd y)+10$: if $p$ was a description of $x$ with condition $y$ with respect to $U_0$, then $\textbf{C}_0p$ is a description of $x$ with condition $y$ with respect to $U$.

The second line ensures that $\KS_U(x\cnd y)$, for each $y$, exceeds $\KS_U(x\cnd \varepsilon)$ at most by $10$. Note that we have $U$ in the right hand side, so this definition is recursive.

The third line implies that $\KS_U(x\cnd x)\le 10$ for every $x$. 

Similarly, the two next lines guarantee that $\KS(f_1(x,y)\cnd y)\le \KS(x\cnd y)+10$ and $\KS(x\cnd y)\le \KS(x\cnd g_1(y))+10$. Let us check it for the first line. Assume that some $p$ is the shortest $U$-description of $x$ given $y$. Then $U(p,y)=x$, and $U(\textbf{C}_3p,y)=f_1(U(p,y),y)=f_1(x,y)$, so $\textbf{C}_3p$ is a $U$-description of $f_1(x,y)$ with condition $y$, and $\KS_U(f_1(x,y)\cnd y)\le |p|+10=\KS_U(x\cnd y)+10$.

Next line: assume that some $p$ is the shortest $U$-description of $x$ with condition $g_1(y)$, i.e., $U(p,g_1(y))=x$. Then $U(C_4p, y)= U(p,g_1(y))=x$, so $\KS_U(x\cnd y)\le |p|+10 = \KS_U(x\cnd g_1(y))+10$.

Now we have to consider the formula for the complexity of pairs in both directions. For that we need some prefix-free encoding for natural numbers and strings. It is easy to check that there exists a prefix-free encoding $n\mapsto e(n)$ for natural numbers (this means that none of $e(i)$ is a prefix of another one) such that $|e(n)|\le 2\log_2 n + 2$.  (We assume in our formulas that $\log 0 =0$, not $-\infty$.) For example, one may note that 
\[
\sum 2^{-2\lfloor\log n\rfloor +2} = 1\cdot\frac{1}{4}+1\cdot\frac{1}{4}+2\cdot\frac{1}{16}+4\frac{1}{64}+\ldots \le 1
\]
and use Kraft--Chaitin's lemma. Then we get the prefix-free encoding for strings (denoted by the same letter $e$):
\[
e(x) = e(|x|)x
\]
such that $|e(x)|\le |x| + 2 \log |x|+ 2$.

Now the formula for pair complexity in one direction is easy: we let
\[
U(\textbf{C}_{300} e(p) q, z)=\langle U(p,z), U(q,\langle U(p,z),z\rangle) \rangle.
\]
If $p$ is a description for $x$ given $z$, we have $U(p,z)=x$. If $q$ is a description of $y$ given $\langle x,z\rangle$, then  $U(q,\langle x,z\rangle)=y$. Then
\[
U(\textbf{C}_{300} e(p) q, z)=\langle x, U(q,\langle x,z\rangle) \rangle = \langle x, U(q,\langle x,z\rangle)\rangle = \langle x,y\rangle,
\]
so $\KS_U(\langle x,y\rangle\cnd z)\le |\textbf{C}_{300}e(p)q|\le 10+|e(p)|+|q|\le |p| + 2\log |p| + |q| +20$.

For the other direction we need to use the Kleene fixed point theorem and the corresponding type of recursive definitions: we describe $U$ that uses the program for $U$ as if it were known. 

  Recall the proof of the formula. For simplicity we consider the unconditional case without $z$, for the general case condition $z$ is added everywhere. We let $n$ be $\KS(x,y)$ and then consider the set of all pairs $\langle x,y\rangle$ that have complexity at most $n$. Among them we consider pairs that have the first coordinate equal to $x$ (and some other second coordinate). Assume that the number of those pairs is between $2^{m}$ and $2^{m+1}$. Then the complexity of $y$ given $x$ does not exceed $m+O(\log n)$, since $y$ can be specified by $x$, $n$, and the ordinal number of $y$ in the enumeration of selected pairs. On the other hand, the complexity of $x$ does not exceed $n-m+O(\log n)$, since $x$ can be specified by $n$, $m$ and the ordinal number of $x$ in the enumeration of all $x$ such that there are at least $2^{m}$ selected pairs with the first component $x$ (there are at most $2^{n-m+1}$ strings $x$ with this property). 

To implement this argument, we need to have to include in $U$ the two encodings used in the argument. Namely, 
\begin{itemize}
\item we let $U(\textbf{C}_{301} e(n)e(m)u, z)$, where $n\ge m$ are natural numbers, $u$ is a binary string considered as a binary representation of a number, and $z$ is a condition, to be the $u$th element in the enumeration of all $x$ such that there exists at least $2^{m}$ strings $y$ such that $\KS_U(\langle x,y\rangle\cnd z)\le n$;
\item we let $U(\textbf{C}_{302} e(n)u, \langle x,z\rangle)$ to be the $u$-th element in the enumeration of all $y$ such that $\KS_U(\langle x,y\rangle\cnd z)\le n$
\end{itemize}
Note that this construction assumes that the program for  $U$ is known when we define $U$ (so we can enumerate the objects in question in some natural order), and the Kleene fixed-point (recursion) theorem is needed to justify this kind of definition. (It the previous constructions we used only values of $U$ on smaller arguments, so it is was just a recursive algorithm.)

  Now, if $\KS_U(\langle x,y\rangle\cnd z)=n$ and the number of $y'$ such that $\KS(\langle x,y'\rangle\cnd z)\le n$ is in the interval $(2^{m},2^{m+1}]$ (note that $m\le n$), then $x$ has description (with condition $z$) of size $|C_{301}|+e(n) + e(m) + (n-m+1) \le n-m+4\log n +5$ and $y$ has description (with condition $\langle x,z\rangle$) of length $|C_{302}|+e(n)+m+1\le m+2\log n + 3$. The sum of lengths is, therefore, at most $20+n+6\log n+8$, and this is a bit better than required.
\end{proof}

\begin{question}
What kind of results gives this approach for the combinatorial applications like aggregation of opinions? 
 \end{question}
 
 \begin{center}
      \large $*$ \quad $*$ \quad $*$
 \end{center}
 
 The author are grateful to all people who insisted on the importance of more precise measuring of complexity, especially to Yury Gurevich, Vladimir A.~Uspensky, Alexey Semenov and Modest Waintswaig.


\begin{thebibliography}{9}

\bibitem{solomonoff1964-1}
R.J.~Solomonoff, A Formal Theory of Inductive Inference. Part I. \emph{Information and Control}, \textbf{7}, 1--22 (1964). 

\bibitem{kolmogorov1965}
A.N.~Kolmogorov, Three approaches to the quantitative definition of information, \emph{Problemy Peredachi Informatsii}, \textbf{1}(1), 3--11. [In Russian. Translated in \emph{International Journal of Computer Mathematics}, \textbf{2}, 157--168 (1968).]

\bibitem{chaitin1966}
Gregory J.~Chaitin, On the Length of Programs for Computing Finite Binary Sequences, \emph{Journal of the ACM}, \textbf{13}(4), 547--569 (1966)

\bibitem{chaitin1969}
Gregory J.~Chaitin, On the Length of Programs for Computing Finite Binary Sequences: Statistical Considerations, \emph{Journal of the ACM}, \textbf{16}(1), 145--159 (1969)

\bibitem{levin1977}
L.A.~Levin, On a concrete method of assigning complexity measures, \emph{Soviet Math. Dokl.}, \textbf{18}(1977), No.~3, 727--731.

\bibitem{chaitin1977}
G.J.~Chaitin, \emph{Algorithmic information theory}, Cambridge Tracts in Theoretical Computer Science 1, Cambridge University Press, 1987.

\bibitem{muchnik-positselsky2002}
A.~Muchnik, S.~Positselsky, Kolmogorov entropy in the context of computability theory, \emph{Theoretical Computer Science}, \textbf{271}(1--2), 15--35 (2002), \url{https://doi.org/10.1016/S0304-3975(01)00028-7}.

\bibitem{durand-levin-shen2007}
B.~Durand, L.~Levin, A.~Shen,  Complex tilings, \emph{Journal of Symbolic Logic}, \textbf{73}(2), 593--613 (2007), see also~\url{https://arxiv.org/abs/cs/0107008}.

\bibitem{chaitin-collection2007} C.~Calude, editor,  \emph{Randomness \& Complexity, from Leibniz to Chaitin}, World Scientific, Singapore (2007), 468 pp., \url{https://doi.org/10.1142/6577}

\bibitem{delahaye-zenil2007}
J.-P.~Delahaye, H.~Zenil, \emph{On the Kolmogorov--Chaitin Complexity for short sequences}, version 1 ( 2007), \url{https://arxiv.org/pdf/0704.1043v1}; version 5 (2010), \url{https://arxiv.org/pdf/0704.1043v5}, also in~\cite[Chapter 6, 123--129]{chaitin-collection2007}.

\bibitem{tromp2007}
John Tromp, Binary Lambda Calculus and Combinatory Logic. In~\cite[Chapter 14, p.~237--260]{chaitin-collection2007}.

\bibitem{tromp}
John Tromp, \emph{Functional Bits: Lambda Calculus based Algorithmic Information Theory}, \url{https://tromp.github.io/cl/LC.pdf} (2023), see also \url{https://tromp.github.io/cl/cl.html}.

\bibitem{nies2009}
A.~Nies, \emph{Computability and Randomness}, Oxford University Press, 2009, ISBN 978-0-19-923076-1.

\bibitem{downey-hirschfeldt2010}
R.~Downey, D.~Hirschfeldt, \emph{Algorithmic Randomness and Complexity}, Springer, 2010, \url{http://link.springer.com/10.1007/978-0-387-68441-3}

\bibitem{gurevich-passmore2012}
Yuri Gurevich, Grant O.~Passmore, Impugning Randomness, Convincingly, \emph{Studia Logica}, \textbf{100}~(1/2), 193--222 (2012), url{https://www.jstor.org/stable/41475223}.

\bibitem{suv}
A.~Shen, V.A.~Uspensky, N.~Vereshchagin, Kolmogorov Complexity and Algorithmic Randomness, American Mathematical Society, 2017, \url{https://www.lirmm.fr/~ashen/kolmbook-eng-scan.pdf}

\bibitem{shen2018} 
A.~Shen, \emph{Compressibility and probabilistic proof}, \url{https://arxiv.org/abs/1703.03342}. (Short version in:  \emph{Unveiling Dynamics and Complexity - 13th Conference on Computability in Europe, CiE 2017, Turku, Finland, June 12--16, 2017, Proceedings}, Lecture Notes in Computer Science, \textbf{10307}, 101--111, Springer, 2017.

\bibitem{li-vitanyi2019} Ming Li, Paul Vit\'anyi, \emph{An Introduction to Kolmogorov Complexity and Its Applications}, 4th edition, Springer, 2019. (Previous editions: 1993, 1997, 2008.)

\bibitem{shen2024}
A.~Shen, Kolmogorov Complexity as a Combinatorial Tool, \emph{Twenty Years of Theoretical and Practical Synergies --- 20th Conference on Computability in Europe, CiE 2024, Amsterdam, The Netherlands, July 8--12, 2024, Proceedings}, Lecture Notes in Computer Science, \textbf{14773}, 27--31.

\bibitem{kss2025}
A.~Kozachinskiy, A.~Shen, T.~Steifer, \emph{Optimal bound for dissatisfaction in perpetual voting}, \url{https://arxiv.org/abs/2501.01969}. Short version: \emph{AIII-25 Proceedings}, AAAI Press, 13977--13984,  \url{https://doi.org/10.1609/aaai.v39i13.33529}.

\end{thebibliography}
\end{document}